\documentclass{article}
\usepackage[top=1.75in,left=1.75in,footskip=0.95in]{geometry}

\usepackage{times}
\usepackage{verbatim}
\usepackage{amssymb}
\usepackage{amsmath}
\usepackage{amsfonts}
\usepackage{color}
\usepackage[pdftex]{graphicx} 
\usepackage{graphicx,graphics}
\usepackage{float}
\usepackage[right]{lineno}
\usepackage{cite}
\usepackage{textcomp,marvosym}
\usepackage{changepage,nameref}
\usepackage{subcaption}

\usepackage[utf8x]{inputenc}

\usepackage{textcomp,marvosym}

\usepackage{microtype}
\DisableLigatures[f]{encoding = *, family = * }

\usepackage[table]{xcolor}
\usepackage{todonotes}

\usepackage{array}

\newcolumntype{+}{!{\vrule width 2pt}}

\newlength\savedwidth

\newcommand{\nor}[1]{\left\lVert {#1} \right\rVert}
\newcommand{\scal}[2]{\left\langle{#1},{#2}\right\rangle}

\DeclareMathOperator*{\argmin}{\arg\!\min}

\newcommand{\E}{\mathbb{E}}

\newcommand{\R}{\mathbb{R}}

\newtheorem{theorem}{Theorem}
\newtheorem{definition*}{Definition}

\newtheorem{remark}{Remark}

\newtheorem{proof}{Proof}
\begin{document}
\vspace*{0.3in}

\begin{flushleft}
{\Large
\textbf\newline{A computational model for grid maps in neural populations} 
}
\newline
\\
Fabio Anselmi\textsuperscript{1,2,*},
Micah M. Murray\textsuperscript{3-6},
Benedetta Franceschiello\textsuperscript{3,4}
\\
\bigskip
\textbf{1} Center for Brains, Minds, and Machines, MIT, Cambridge, MA, USA
\\
\textbf{2} Laboratory for Computational and Statistical Learning (LCSL) and IIT, Genova, Italy \\
\textbf{3} The LINE (Laboratory for Investigative Neurophysiology), Department of Radiology, University Hospital Center and University of Lausanne, Lausanne, Switzerland \\
\textbf{4} Ophthalmology Department, Fondation Asile des Aveugles and University of Lausanne, Lausanne, Switzerland \\
\textbf{5} EEG Brain Mapping Core, Center for Biomedical Imaging (CIBM), Lausanne, Switzerland \\
\textbf{6} Department of Hearing and Speech Sciences, Vanderbilt University, Nashville, TN, USA \\
\bigskip
* anselmi@mit.edu
\end{flushleft}

\section*{Abstract}
Grid cells in the entorhinal cortex, together with head direction, place, speed and border cells, are major contributors to the organization of spatial representations in the brain. In this work we introduce a novel theoretical and algorithmic framework able to explain the emergence of hexagonal grid-like response patterns from head direction cells' responses. We show that this pattern is a result of minimal variance encoding of neurons. The novelty lies into the formulation of the encoding problem through the modern Frame Theory language, specifically that of equiangular Frames, providing new insights about the optimality of hexagonal grid receptive fields. The model proposed overcomes some crucial limitations of the current attractor and oscillatory models. It is based on the well-accepted and tested hypothesis of Hebbian learning, providing a simplified cortical-based framework that does not require the presence of theta velocity-driven oscillations (oscillatory model) or translational symmetries in the synaptic connections (attractor model). We moreover demonstrate that the proposed encoding mechanism naturally explains axis alignment of neighbor grid cells and maps shifts, rotations and scaling of the stimuli onto the shape of grid cells' receptive fields, giving a straightforward explanation of the experimental evidence of grid cells remapping under transformations of environmental cues. 
 


 
\section*{Author Summary}

Since their discovery in 2005 grid cells have played a key-role in understanding how different species' brains dynamically represent an animal's position in space. 
Despite more then a decade of interest from a large number of investigators, a universally accepted model of how grid cells receptive fields emerge is still lacking.
In this study we provide a new and simple theoretical and computational framework to explain how grid cells could possibly emerge from the firing activity of  head direction cells. 
We propose a novel formulation of the encoding problem through the modern Frame Theory language, providing new insights about the optimality of hexagonal grid receptive fields and  overcoming some crucial limitations of the current attractor and oscillatory models.
Moreover, we demonstrate that this same encoding strategy can generalize from spatial to more abstract information.


\section*{Introduction}

Grid cells in the entorhinal cortex efficiently represent an animal's spatial position using an hexagonal symmetric code \cite{Moser2005,Burak2009}. Mathematical models have been developed to explain the emergence of such surprisingly regular firing activity~\cite{McNaughton2006,Fuhs2006,Burgess2007,Hasselmo2007,Blair2007,Kropff2008}. However, the problem is far from being solved, and many questions remain open \cite{Renart2003,Yartsev2011,Schmidt2013,Heys2013}. From the modelling point of view, two main mechanisms have been proposed to generate the hexagonal periodic activity: oscillatory interference \cite{Burgess2007,Orchard2013} and continuous attractor dynamics \cite{McNaughton2006,Fuhs2006}. First, we address briefly the main ideas underlying these models. 

In oscillatory models, grid cells' patterns emerge from the interference between theta oscillations of velocity-modulated cells \cite{Burgess2007,Hasselmo2007}. Experimental results in \cite{Burgess2012} have identified a class of cells, named band cells, that fire at specific spatial periodicity; the interference of three cells of this kind, whose wave vectors' orientations differ by multiples of $120$ degrees, leads to hexagonal grid-type interference patterns. 
The key assumptions of oscillatory models have been experimentally challenged. Theta oscillations have not been observed in fruit bats \cite{Yartsev2011} or macaque monkeys \cite{Killian2012}, despite
robust grid cell activity having been recorded in both species.
Furthermore, whole-cell recordings from head-fixed mice moving at different velocities, showed minimal correlation between the mouse's velocity and the amplitude or periodicity of theta oscillations \cite{Domnisoru2013,Schmidt2013}.

The core idea of continuous attractor models explains the regularity of the grid firing activity as an attractor state generated by symmetrical recurrent interactions between grid cells \cite{McNaughton2006,Fuhs2006}. A major weakness of this class of models is that it requires an unrealistically high degree of translational symmetry in the strength of the connections among neurons: neurons at equal distance should connect with equal strength. However, real neuronal populations are affected by noise and randomness and therefore break this symmetry and the grid regularity \cite{Renart2003}. Alternative models based on single-cell firing, adaptation, slowly varying spatial inputs, or, more recently, on deep reinforcement learning have been proposed in \cite{Kropff2008,wiskott2007,banino2018, Ger2017}.\\ 
The model we propose has a number of advantages with respect to those mentioned above. For clarity we list the novel contributions of our work:
\begin{itemize}

\item The model is based on the well-accepted and tested hypothesis of Hebbian learning, \cite{Hebb1949}, and none  of the complications of the interference and attractors models are needed, like the presence of theta-velocity driven oscillations or translational symmetries in the synaptic connections. 

\item We explain how grid cells could emerge from the head direction population firing activity, giving a  novel principled derivation of the hexagonal grid shape using signal analysis, in particular frame theory.

\item  We explain the experimental phenomenon of neighbor grid cells axes alignment. 
 
\item We show how shifted, rotated and scaled grid cells' receptive fields naturally remap, given transformed visual landmarks~\cite{Moser2005,Fyhn2007}.

\item 
We sketch a theoretical framework for the otherwise puzzling experimental findings in~\cite{Costantinescu2016} where the authors show how grid cells may play a role in the organization of "conceptual" spaces. 

\end{itemize}

\section*{Results}
\subsection*{Model description and predictions}


The model is based on three assumptions. By analogy with the Hubel and Wiesel simple-complex cells computation in the primary visual cortex ~\cite{Hubel1965,Hubel1968}, we propose grid cells to emerge from a linear sum of "simple cells", that we will identify with head direction cells, whose receptive fields are learned from a collection of neuronal inputs with \emph{stationary} second-order stimulus statistics ($\mathbf{H}1$). 
In other words, we assume that the encoding of the objects' movements at the level of the entorhinal cortex (the upstream neuronal responses) obeys a statistic that does not differ significantly from that of natural images (which is approximately stationary \cite{field1999}). Deep connections between visual recognition task and enthorinal cortex has been suggested in \cite{bica2019}. 
We also assume that each neuron computes a response that is the scalar product between the input and its synaptic weights i.e.
\begin{equation}
r_i(\mathbf{x})=\scal{\mathbf{x}}{\mathbf{w}_i}
\end{equation} 
with $\mathbf{x}$ the input image function and $\mathbf{w}_i$ the synaptic weights function of neuron $i$.
In the following, we will fix $\mathbf{x}$ and omit to write the dependence on $\mathbf{x}$. \\
Furthermore, we assume that the synaptic weights are updated following Oja's rule, derived as a the first order expansion of a normalized Hebbian rule, \cite{oja1992principal}, ($\mathbf{H}2$). The normalization assumption is plausible, because normalization mechanisms are widespread in the brain \cite{carandini2018}. 
The original paper of Oja \cite{oja1982simplified} showed that the weights of a neuron updated according to this rule will converge
to the top principal component (PC) of the neuron's past input, i.e. to an eigenfunction of the input's covariance. Plausible modifications of the rule, involving added noise or inhibitory connections with similar neurons, yield additional eigenfunctions \cite{oja1992principal}. Thus, this generalized Oja's rule can be regarded as an online algorithm to compute the principal components of incoming streams of input; in our case, the stationary neuronal responses of simple cells.
Our last and most important assumption is that the neural population's goal is to encode a variation of its input, in this case the position, with maximal precision \cite{Deneve1999} $(\mathbf{H}3)$: neuronal responses are noisy, and thus repeated, equal stimuli can produce different outputs. $(\mathbf{H}3)$ tells us that the population coding aims to minimize the variance of the responses.\\
\noindent
The first important consequence of hypotheses $(\mathbf{H}1,\mathbf{H}2)$ is that the synaptic weights of the neuronal population are tuned to Fourier functions i.e. 
\begin{equation}\label{Fourierc}
w(\mathbf{k},\mathbf{\xi})=e^{I\mathbf{k}^{T}\mathbf{\xi}},\;\mathbf{k},\mathbf{\xi}\in\R^2
\end{equation}
where $I$ is the imaginary number and $\mathbf{k}$ is the two-dimensional frequency vector.
This follows from the stationarity of neuronal input i.e. the fact that the associated covariance matrix is diagonalized by Fourier functions. A consequence of Oja's rule is that those are also the learned neuronal weights. The relative change of position of the objects in the scene (due to the animal navigating in the environment) is modelled in a first approximation as covariant translations at the level of the highly processed input of the enthorinal neurons i.e. :
\begin{equation}\label{transl}
T_\mathbf{y}\mathbf{x}(\mathbf{\xi})=\mathbf{x}(\mathbf{\xi}-\mathbf{y}),\mathbf{y}\in\R^{2};
\end{equation}
where $T_\mathbf{y}$ is the translation operator. The response of a $N$ neurons population encoding the position of stimulus $\mathbf{y}$ will be:
\begin{equation}
\mathbf{r}(\mathbf{y})=(r_1(\mathbf{y}),\cdots,r_N(\mathbf{y}))
\end{equation}
with $r_i(\mathbf{y})=\scal{T_{\mathbf{y}}\mathbf{x}}{e^{I\mathbf{k}^{T}_i\mathbf{\xi}}}$.
Upon a change in the observer head direction of an angle $\theta$, implemented by the rotation matrix $R_{\theta}$ , the simple cells response changes as:
$$
r_i(\mathbf{y},\theta)=\scal{R_{\theta}T_{\mathbf{y}}\mathbf{x}}{e^{I\mathbf{k}^{T}_i\mathbf{\xi}}}=e^{I\mathbf{k}^{T}_i\mathbf{y}}[R_{\theta}c(x)]_{i}
$$
where $c_i(\mathbf{x})$ are the Fourier coefficients of $\mathbf{x}$ with respect to the frequencies $\mathbf{k}_i$ and $e^{I\mathbf{k}^T_i\mathbf{y}}$ are the phase factors due to the translation covariance of the Fourier transform.
Interestingly, the information about the rotation angle  $\theta$ and the translation $\mathbf{y}$ is contained into two different parts of the Fourier transform, respectively a phase factor (for the translation) and the rotation of the untransformed Fourier coefficients. 
The previous observations allow us to perform an analogy between simple cells in our model and HDs in the hippocampus.
Our choice is supported the by anatomical and experimental evidence that HDs directly input grid cells and disruption or disturbance of HDs signal destroy grid cells hexagonal firing activity \cite{Winter2015}.
\\
In the following we will focus on the emergence of grid cells' receptive fields, in particular on how they can be derived from optimality of the position information contained in the  phase factors of the HDs cells' response.\\
The simplest model of a \emph{"complex" grid cell} aggregates the responses of simple head direction cells by summation:
\begin{equation}
r(\mathbf{y})=\sum_{i=1}^{N} r_i(\mathbf{y}) = \scal{T_{\mathbf{y}}\mathbf{x}}{\sum_{i=1}^{N}\mathbf{w}_i}=\sum_{i=1}^{N} c_i(\mathbf{x})e^{I\mathbf{k}^T_i\mathbf{y}} \label{response}
\end{equation}
 The phases, as in \cite{Orchard2013}, encode the information about the observer position.

In general, each single simple cell's response can be considered as a random variable subject to Gaussian noise:
\begin{equation}\label{noise}
r_i(\mathbf{y})\rightarrow r_i(\mathbf{y})+\sigma_i.
\end{equation}
In the following, we suppose the noise at each simple cell to be the same and uncorrelated, i.e.:
$(1/\sigma^2)\mathbf{I}$. 
Assuming $(\mathbf{H}2)$ and $(\mathbf{H}3)$, the question that follows is how many neurons and which set of frequencies $\{\mathbf{k}_i\}$ are best to encode the animal's position $\mathbf{y}$ with maximal precision.\\ 
A lower bound on any possible unbiased estimator of the random variable $\mathbf{y}$ is given by the Cramer-Rao bound (\cite{Kay1993}, see Materials and Methods \ref{app:secFisher} for more details). The bound reads:
\begin{equation}\label{CR}
\nor{Cov(\mathbf{y})}\geq \nor{F^{-1}(\mathbf{y})}
\end{equation}
where $\mathbf{F}$ is the so-called \emph{Fisher information} matrix, $Cov$ is the covariance matrix of the neuronal responses, and $\nor{\cdot}$ is a matrix norm.
Intuitively, $\mathbf{F}$ measures the amount of information that the encoding population carries about the random variable $\mathbf{y}$. 
The main theoretical result of the paper is that the lower bound for the right hand side of Eq. \ref{CR} is achieved when the set of frequency vectors $\{\mathbf{k}_i\}$ form a so-called \emph{Equiangular frame}. An equiangular frame is a set of vectors such that each pair of vectors form the same angle between them; i.e. 
\begin{equation}
\mathbf{k}^T_p \mathbf{k}_q=\cos(\alpha),\;\forall\;p,q; \;\alpha\in [0,2\pi], \alpha \mbox{ constant}. 
\end{equation}
We proved that the lower bound is achieved in dimension two when the number of neurons is $N=2$, yielding the Orthonormal frame. For $N=3$ we obtain the so-called Mercedes-Benz frame, composed by vectors whose orientations differ by $60^\circ$ degrees (see Fig. \ref{fig:frames}).
\begin{figure}[H]
    \centering
    \includegraphics[width = 0.7\textwidth]{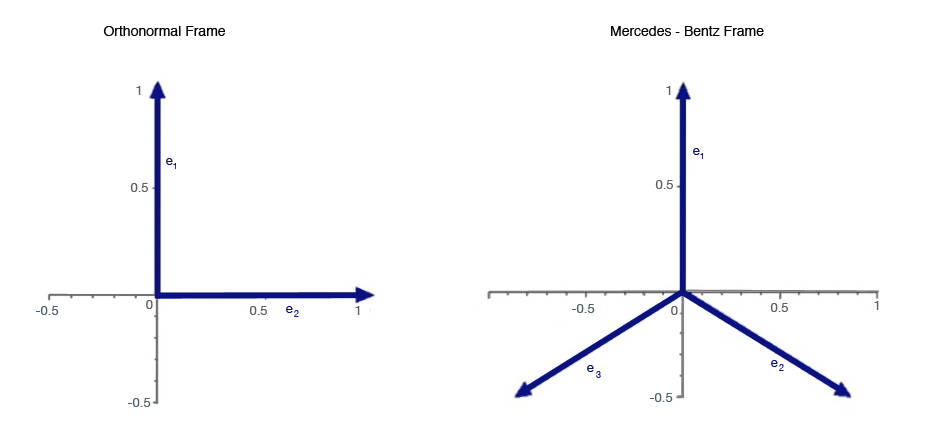}
    \caption{Orthonormal (left) and Mercedes-Benz frame (right) in dimension two. Note how the vectors in the Mercedes-Benz frame are separated by $120^{\circ}$ one from the other.}
    \label{fig:frames}
\end{figure}
Clearly it is biologically implausible that grid cells receives input from so few HDs. The activity of one simple cell unit in our model summarizes that of a whole population of cells with the same preferential orientation $\mathbf{k_{i}}$: summing over multiple simple cells responses sensitive to the same orientation keeps the value of eq. \eqref{response} unchanged up to an overall constant factor.\\  
\noindent
More formally:
\begin{theorem}\label{thmain}
Given the hypotheses $H(1,2,3)$ the minimum variance position encoded by a set of $N$ neurons corrupted by Gaussian uncorrelated noise is achieved when $N=3$ and the set of frequency vectors is:
$$
f=\{\mathbf{k}_1,\mathbf{k}_2,\mathbf{k}_3\}=\{(\cos(2\pi j/3),\sin(2\pi j/3)),j=1,\cdots,3\}
$$
or for $N=2$ the orthonormal frame.
\end{theorem}

\begin{figure}[H]
\includegraphics[width=1\textwidth]{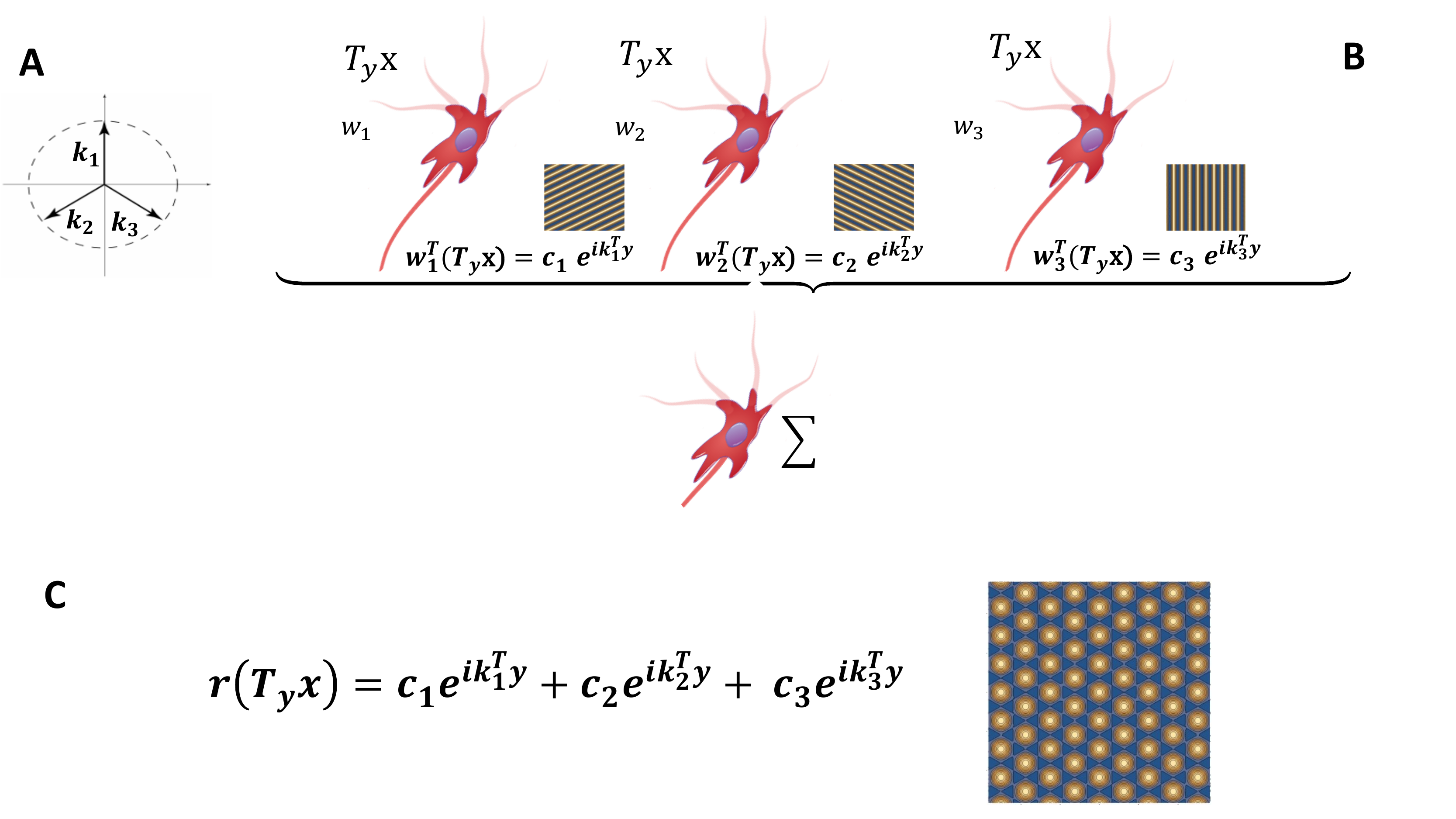}
\caption{The image shows how a grid-like cells pattern arises from the interference of planar waves responses. (A) The Mercedes-Benz frame is constituted by the equiangular vectors $k_1, k_2, k_3$, whose directions are along the angles $\theta = \pi/2, -\pi/6, -5\pi/6$. (B) Three neurons with input stimulus $x$, translated by a vector $y$, $T_{y}x$ and their receptive fields $w_1,w_2,w_3$ i.e. three planar waves in equiangular directions.  (C) linear sum of the three neurons' responses (indicated in B as $\Sigma$) resulting in the grid-like hexagonal pattern.} \label{figure_grid_cell}
\end{figure}

Bearing in mind the previous observation, our result states that the best encoding of position is achieved when the "complex" grid-like cell is aggregating the responses of two or three neurons (or similarly-tuned neural populations) whose neuronal weights are Fourier functions with \emph{equiangular frequencies} in the frequency space.
Suppose now the neurons' weights have been learned.
The response in eq. \eqref{response} produces output in terms of an interference pattern of two or three planar waves that is consistent with a square or hexagonal grid, respectively (see Figure \ref{figure_grid_cell}).
Importantly, grid responses are robust to noise: their stability follows from the stability to noise of the simple head direction cells receptive fields. This is guaranteed by the proven noise robustness of  Oja's online learning algorithm, \cite{Kar1994}.\\
The novelty of our contribution w.r.t. existing literature is summarised in the following remarks.\\
\begin{itemize}
\item Our model may resemble the interference model of \cite{Burgess2007}. However, it is important to point out one crucial difference: the oscillations interference pattern in our model are due to the shape of the learned receptive fields of the HDs and not to the theta oscillations in the hippocampal circuit.
Therefore, band oscillations in Figure \ref{grid_fig} are not the controversial ones described in \cite{kupric2014}.\\
\item In the same vein, although PCA is used in our model to derive the shape of the HDs' receptive fields, its role is completely different from the one in \cite{Dordek2016}, where PCA is used to derive the shape of the grid cells' responses. \\ 
\item Optimality of grid cells' hexagonal receptive field is here derived in a novel way w.r.t. e.g. in \cite{fiete2008,Mathis:2012,vago2018,caplan2018}. In fact, we used instruments formalised within the Frame theory and the interesting properties of the so-called Mercedes-Benz frame (see e.g. \cite{Framet}).\\
\end{itemize}
\subsection*{Two phase application of Oja's learning rule}
In this section we explain how hexagonal grid receptive fields emerge from a two-phase learning process.

\paragraph*{First phase: "simple" HD cells learning.} A collection of cyclical translations in the $(x,y)$ directions of natural images is used as input stimuli to compute the "simple" cell profile of activation. Next, the principal components of these activation profiles are extracted, diagonalizing the covariance matrix of the input data. The second order statistic of the input, i.e. the covariance matrix, is clearly  stationary.
Note that the stationarity is independent from the nature of the stimulus and it crucially depends on the more abstract notion of transformation (in this case translations).
Under the assumption of Oja's rule, this mechanism simulates the learning phase of a "simple" HD cell.\\
\noindent
An example of the learned receptive fields is shown in Fig \ref{grid_fig} A. As expected, they are Fourier components with different directions.
\paragraph*{Second phase: "complex" grid cell learning.} The second step entails aggregating the responses of simple HD cells. A collection of cyclical translations of a test image is used to calculate the aggregation vector $J$ according to the minimization problem (see Figure \ref{grid_fig} and Materials and Methods, \textit{Algorithmic formulation} for the algorithm details).

\begin{figure}[H]
\includegraphics[width=1\textwidth]{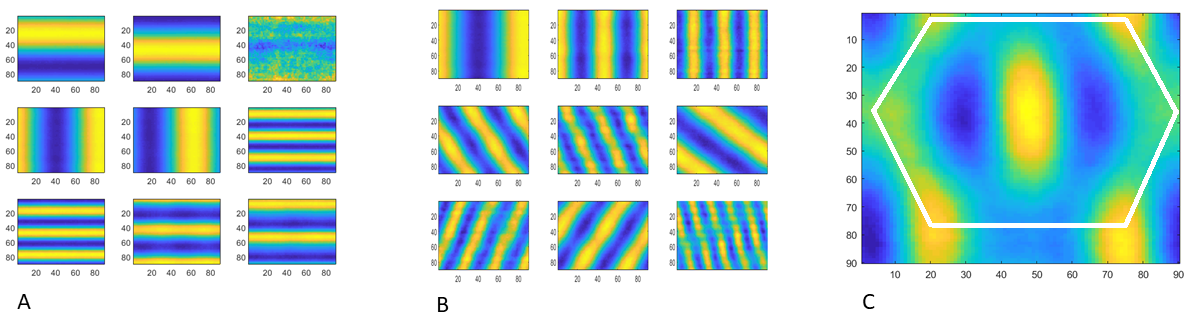}
\caption{(A) Simple head direction cells receptive fields, obtained through the first phase learning, from stationary stimuli. (B) Head direction receptive fields, selected through variance minimization of the estimated position. (C) Superposition of equiangular patterns selected from figure (B), for N=3.}
\label{grid_fig}
\end{figure}

The results displayed in Figure \ref{grid_fig} show that,
as a result of variance minimization of the position estimate, waves with $120^{\circ}$ angular distance are selected (B). Their superposition have a grid-like shape (C). Although the algorithm provides angular directions predicted by the theorem, not all superpositions produce grid-like patterns.
This is due to the different frequencies of the Fourier receptive fields.
For a fixed frequency, the selected receptive fields sum to produce a grid-like interference pattern. 

The mechanism underpinning the combination of receptive fields of the same frequency relies on the nature of principal component decomposition: the first eigen-component is an oscillating wave whose frequency depends on the strongest oscillating component in the stimuli (see \cite{AapoBook}, pg 120). Thus, cells with the same RF size will be tuned to the same frequency.  
Moreover, cells in the same spatial neighbourhood, receiving the same input, will be tuned to similar wave orientations vectors. This explains a salient aspect of grid cells phenomenological behavior: neighboring grids cells have aligned orientations of their axes (i.e. the same orientations of the hexagonal axes). 

\subsubsection*{Grid remapping by changing environmental cues}

Experimental evidence shows that changes in environmental cues are matched by a transformation in the animal grid cells responses \cite{Fyhn2007}. For example a rotation of the main visual cues in the environment results in a rotation of the grid cells orientation field.\\
This aspect can be readily explained by our model. For example, if an environmental cue is rotated by an angle $\theta$ the grid rotates accordingly, since
\begin{equation*}
r_i(\mathbf{R}_\theta \mathbf{x})=\scal{\mathbf{R}_{\theta}\mathbf{x}}{\mathbf{w}_i}=\scal{\mathbf{x}}{\mathbf{R}_{-\theta}\mathbf{w}_i}
\end{equation*}
where $\mathbf{x}$ is the input of the HDs. In this case the frequency vectors $\{\mathbf{k}_{i}\}$ will be all rotated by the opposite angle $R_{-\theta}\mathbf{k}_{i}$, with a resulting rotation of the hexagonal grid. Similarly for a scale transformation the frequency vectors  will be rescaled by the scaling factor.

\section*{Discussion}
We successfully showed how hexagonal receptive fields, resembling to those of grid cells, emerge naturally in the spatial encoding framework by requiring minimal variance (maximal precision) of the population encoding together with Oja's learning rule. 
The assumption of a 2-phase simple HD-complex cell type learning adapts properties typically found in the visual cortex to those characterizing the entorhinal cortex. We contend that similarity in the types of learning is plausible, given that the entorhinal cortex integrates visual information while also determining the relative position of the observer navigating the spatial environment.
Importantly, the presented model provide a theoretical framework capable of explaining the experimental evidence that grid cells encode an abstract notion of space, decoupled from the specificity of the sensory inputs. The notion indeed emerges from the mathematical group properties of the objects' transformations, rather than the objects themselves.
More generally, our model would indicate that grid-like coding should manifest whenever the statistics of the neuronal inputs is stationary.
Indeed, the model detailed here provides a mathematical framework able to mimic the emergence of grid-like patterns not only in a spatial encoding scheme (where the considered transformations of the space are translations), but also in a more conceptual encoding scheme (where the transformations are dilations, e.g. \cite{Costantinescu2016}).

\subsection*{``Conceptual" encoding schemes\label{main:sec3}}

The idea that grid-like cells could provide a model to understand "cognitive", in addition to sensory-related brain functions, is not new \cite{moser2013,moser2014}. However, it was not before the work in \cite{Costantinescu2016} that the first experimental evidence was provided. Interestingly, our findings can be applied to outline a theoretical framework for investigating a possible computational model of their experimental evidence.

The stimuli in \cite{Costantinescu2016} are described in a two-dimensional \textit{conceptual bird space}, where the position coordinates are the lengths of both the neck and legs of the bird. In \cite{Costantinescu2016} the authors show an hexagonal grid-like pattern, while testing conceptual associations with functional magnetic resonance imaging (fMRI). For simplicity, we model the input space by using the shear group in $2D$ (composed of transformations dilating an image in the $x,y$ directions). Instead of the ratio between the legs of the birds and their necks we can think about the ratio between the base and height of a rectangle that scale in the directions $x$ and $y$, respectively, according to the parameters $(l_1,l_2)=\mathbf{l}\in\R^2$ (see Figure \ref{fig:bird}, C). 
\begin{figure}[H]
    \centering
    \includegraphics[width = 0.7\textwidth]{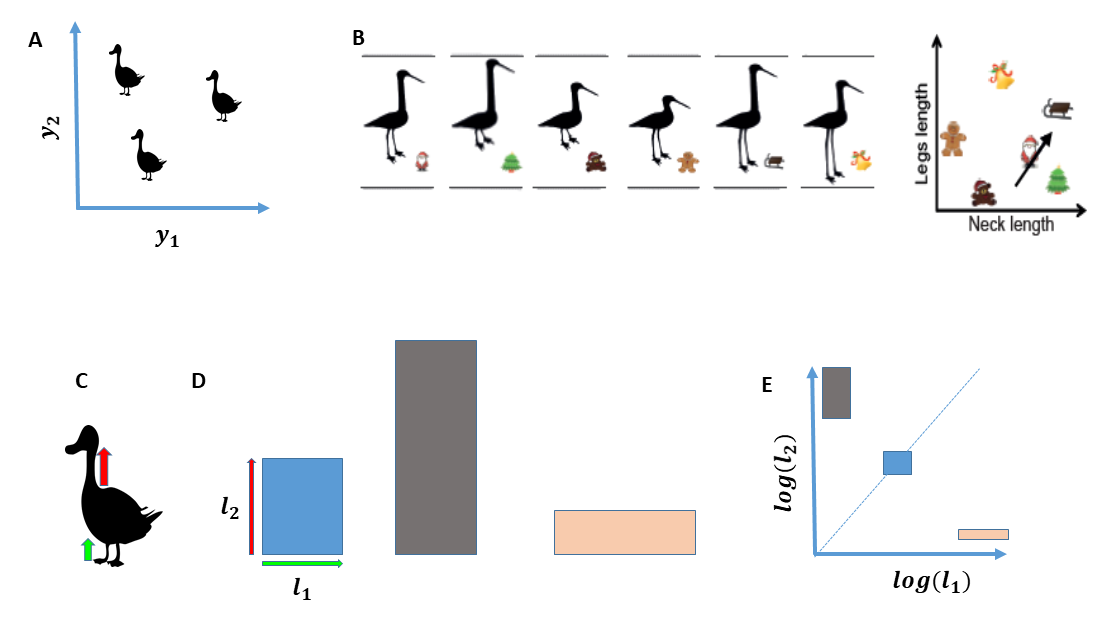}
    \caption{Translations of a bird in space (A). Transformations of a bird (B) and associated points in the `bird space''(from \cite{Costantinescu2016}). Transformations of a rectangle (D) that simplify the bird transformations (C).  The  associated points in the "rectangle space" (E).}
    \label{fig:bird}
\end{figure}
\noindent
The transformation corresponds to Eq. \eqref{transl} where instead of the translator operator the shear operator was used:
\begin{equation}\label{dil}
D_\mathbf{l}\mathbf{x}(\mathbf{\xi})=\mathbf{x}\bigg(\frac{\mathbf{\xi}}{\mathbf{l}}\bigg).
\end{equation}
Note that $D_\mathbf{l}$ indicates the shear operator.
The main idea is to apply our spatial encoding to assess whether the model allows to represent the grid-like conceptual patterns observed in \cite{Costantinescu2016}. 
We stress that similarly to the bird space (where the direction of motion in the abstract "bird space" was determined by the ratio between the neck and the leg lengths), the direction of motion (in our abstract "scale space") is determined by the ratio between the base and height of the rectangle.
It is simple to demonstrate that also in this case the second order statistic of the input is stationary since it again depends only on the nature of the transformation. Therefore, the synaptic weights of the neuronal population are tuned to the eigenfunctions of the shear operator as they were before in the translation case. These eigenfunctions are a generalization of Fourier components to the shear group and have the same form, but in the log-scale coordinates $\log(\mathbf{l})=(log(l_1),log(l_2))$, where $l_1,l_2$ are the scaling factors in the $x$ and $y$ directions (see e.g. \cite{eagleson1992}) :
\begin{equation}
s(\mathbf{l},\mathbf{k})=e^{I\mathbf{k}^{T}\log(\mathbf{l})}.
\end{equation}
The key observation is that in this new coordinates frame (provided by the response of the ''simple cells''), the shear transformations reduce to translations in the $\log(\mathbf{l})$-space since:
\begin{equation}
\log(\mathbf{l}\mathbf{l'})= \log(\mathbf{l})+\log(\mathbf{l'}).
\end{equation}
In this space the eigenfunctions are planar waves as in Figure \ref{figure_grid_cell} A, applying Theorem \ref{thmain}. We can then prove that also in this case the set of frequency vectors $\{\mathbf{k}_i\}$ is the Mercedes-Benz frame or the orthonormal frame. This will produce a square or hexagonal grid in the \textit{shear space}, where the coordinates, instead of the spatial $\mathbf{x,y}$ are the scale coordinates $\mathbf{l}=(l_1,l_2)$ as in Figure \ref{fig:bird}, (E).\\

The results in Theorem \ref{thmain} can be generalized to any Abelian group; the eigenfunctions of the group transformations are the group characters \cite{eagleson1992}. It can be used to predict grid-like cell geometries in higher dimensions. In dimension three, a possible solution corresponds to the vectors associated to the vertices of a tetrahedron. More generally, in dimension $d$ we will end up with the vectors associated to the vertices of a Platonic Solid. However, it should be noted that many other solution configurations might exist that are distinct from the case of $d=2$ analyzed in this paper.


\subsection*{Conclusion and Outlook}

We detailed a computational model able to account for the emergence of hexagonal grid-like response patterns that derives from head direction cells responses and neural sensitivity to the statistics of the input stimuli (i.e. minimal variance encoding). By applying Frame Theory, we provided a novel formulation of the encoding problem within the framework of equiangular frames. Our model overcomes some limitations described for both oscillatory and attractor models. Moreover, we were able to demonstrate that grid-like receptive field patterns persist despite transformations of the environmental cues as well as when more "conceptual" features are considered as input stimuli. Further work will be required to extend our findings to reproduce the experimental evidence showing that the regular pattern of the grid receptive field adapts to different geometries of the environment, distorting its hexagonal regularity \cite{kupric2014, Treves2015,kei2019}. Our main result in Theorem \ref{thmain} predicts the same hexagonal grid for the 3D space of rotations of an object, leading to a series of experiments in the same spirit of \cite{Cheng2018,Kim2019,Jacobs2013} (for spatial encoding) and of \cite{Costantinescu2016} for conceptual encoding possibly tested using magnetoencephalography \cite{stau2019}.


\section*{Materials and Methods}

\subsection*{Fisher information}\label{app:secFisher}

The Cramer-Rao bound (CRB) sets a lower bound on the norm of the covariance operator of any random variable unbiased estimator.
It says:
\begin{equation*}
\nor{Cov(\mathbf{y})}\leq \nor{F^{-1}(\mathbf{y})}
\end{equation*}
where $\mathbf{y}$ is our position variable, $F$ is the Fisher information defined as:
\begin{equation*}
(F(\mathbf{y}))_{k,l}=-\E\Big(\frac{\partial^2 \log L(\mathbf{y})}{\partial y_k \partial y_l}\Big)
\end{equation*}
and $L(\mathbf{y})$ is the likelihood function and the average is over of the measurements of $\mathbf{y}$. 
In our case the likelihood function is, \cite{Yoon1998}:
\begin{equation}\label{LH}
L(\mathbf{y}) = \mathcal{N}exp\Big(-\frac{1}{2}\sum_{i,j=1}^{N} (y_i-r_i(\mathbf{y}))C^{-1}_{ij}(y_j-r_j(\mathbf{y}))\Big).
\end{equation}
where $\mathcal{N}$ is a normalization constant, $C$ is the correlation matrix of the Gaussian noise and $r_i$ is the response of the $i^{th}$ neuron. Under the hypothesis of uncorrelated equal noise, i.e. $\mathbf{C}=(1/\sigma^2)\mathbf{I}$, a direct calculation of Eq.~\eqref{LH} gives:
\begin{eqnarray*}
\mathbf{F}(\mathbf{y})&=&-\E\Big\{\Big(\frac{\partial \mathbf{r}(\mathbf{y})}{\partial \mathbf{y}}\Big)^{\dagger} \mathbf{C} \frac{\partial \mathbf{r}(\mathbf{y})}{\partial \mathbf{y}}\Big\} = -\frac{1}{\sigma^2}\E\Big\{\Big(\frac{\partial \mathcal{F}(T_{\mathbf{y}}\mathbf{x})}{\partial \mathbf{y}}\Big)^{\dagger} \frac{\partial \mathcal{F}(T_{\mathbf{y}}\mathbf{x})}{\partial \mathbf{y}}\Big\},
\end{eqnarray*}
where $\mathcal{F}$ is the Fourier transform. Starting from the following identity:
\begin{equation*}
\frac{\partial\mathcal{F}_i(T_{\mathbf{y}}\mathbf{x})}{\partial \mathbf{y}}= Ie^{{I\mathbf{k}_i}^T\mathbf{y}}\mathbf{k}_ic_{i}(\mathbf{x}),\;\;
\mathbf{c}(\mathbf{x})=\mathcal{F}(\mathbf{x}),
\end{equation*}
we have:
\begin{eqnarray} \label{Fisher}
\mathbf{F}&=&\frac{1}{\sigma^2}\E\Big\{\sum_{i=1}^N \mathbf{k}_i {\mathbf{k}_{i}}^T|c_{i}(\mathbf{x})|^2\Big\}
\frac{1}{\sigma^2}\sum_{i=1}^N \mathbf{k}_i {\mathbf{k}_{i}}^T\E(|c_{i}(\mathbf{x})|^2)\\ 
&=&\frac{1}{\sigma^2} \sum_{i=1}^N\frac{\mathbf{k}_i {\mathbf{k}_{i}}^T}{\nor{\mathbf{k}_{i}}_2^2}=\frac{1}{\sigma^2}\sum_{i=1}^N\mathbf{g}_i {\mathbf{g}_{i}}^T \nonumber 
\end{eqnarray}
where we used the fact that the averaged power spectrum $\E(|c_{i}(\mathbf{x})|^2)\approx \nor{\mathbf{k}_{i}}^{-2}_2$ and we define the unit norm vector $\mathbf{g}_i=\mathbf{k}_i/\nor{\mathbf{k}_i}$.\\
\noindent
The question we address in the next paragraph is: for which set of $\mathbf{k}_i$ is the (CRB) achieved? In other words, we are looking for the values of $\mathbf{k}_i$ for which the neuronal population is providing an estimate of the variable $\mathbf{y}$ with minimal variance.
In particular we consider the following minimization problem:
\begin{equation}\label{minim}
\argmin_{\{\mathbf{k}_i\}_{i=1}^{N}}\nor{\mathbf{F}^{-1}}_{Frob}^2.
\end{equation}

\subsection*{Optimal estimator and connection with frame theory}\label{app:sec2}
In this section we derive the proof of the main result of the paper.  
\begin{theorem}
Under the hypotheses $\mathbf{H}(1,2,3)$ the minimal variance position encoded by a set of $N$ neurons corrupted by Gaussian uncorrelated noise is achieved when $N=3$ and the set of frequency vectors is:
$$
f=\{\mathbf{k}_1,\mathbf{k}_2,\mathbf{k}_3\}=\{(cos(2\pi j/3),sin(2\pi j/3)),j=1,\cdots,3\}.
$$
or when $N=2$ for the set of frequencies forms an orthonormal frame.
\end{theorem}
%
\begin{proof}
Using the fact that $\mathbf{F}$ is semi-positive definite we can decompose it as $\mathbf{F} = \mathbf{V}^{T}\mathbf{\Lambda} \mathbf{V}$ where $\mathbf{V}$ is unitary and $\mathbf{\Lambda}$ is diagonal.
According to the Cramer-Rao bound the variance is bounded from below by the inverse of the Fisher Information.
Calculating the Frobenius norm of the Fisher matrix inverse we have:
\begin{equation}\label{tracen}
\nor{\mathbf{\mathbf{F}}^{-1}}^2 = Tr(\mathbf{V}^{T}(\mathbf{\Lambda}^{-1})^2 \mathbf{V})= Tr((\mathbf{\Lambda}^{-1})^2)= \sum_{i=1}^N\frac{1}{\lambda^2_i}
\end{equation}
where $\lambda_i$ are the eigenvalues of $\mathbf{F}$.\\
\noindent
It is easy to prove that the minimum of eq.\eqref{tracen} is reached when all the eigenvalues are equal i.e.
\begin{equation}\label{framecond}
\mathbf{F}=\lambda \bf{I}
\end{equation}
i.e. the set $\mathbf{k_i}$ form a tight frame.
When $N$ is equal to $2,3$, the only solutions are (see \cite{Goyal2001}, pg 210): for $N=2$, the orthogonal frame, for $N=3$ the so called Mercedes-Benz frame (or any rotated version of them):
\begin{equation*}
\mathbf{k}_j=\Big(\cos(\frac{2\pi j}{3}),\sin(\frac{2\pi j}{3})\Big)\;\;\;j=1,\cdots,3.
\end{equation*}
\end{proof}
\begin{remark}
In the three dimensional case a similar result still holds: a possible solution corresponds to the vectors associated to the vertices of a tetrahedron. More in general in dimension $d$ will be the vectors associated to the vertices of a Platonic Solid. However the uniqueness of the solution is guaranteed only in dimension two.
\end{remark}

\subsection*{Algorithmic formulation}

In this article we suppose a 2-phases learning process:
\begin{enumerate}
\item[(1)] Learning of the Fourier components by simple HD cells using Oja's synaptic updating rule;
\item[(2)] Learning of the selection of simple HD cells, performed by the complex cell that minimize, according to the Cramer-Rao bound, the norm of the inverse of the Fisher Information. 
\end{enumerate}
Solving phase $1$ simply correspond to the extraction of 
the principal components of neural input.
As for phase $2$ the minimization problem for the complex cell is:
\begin{equation}\label{algo2}
\argmin_{\{\mathbf{k_i}\}_{i=1}^N}\;\; \nor{F^{-1}}^2\;\;\;\textrm{with}\;\;N\;\;\;\textrm{minimal}.
\end{equation}
As we saw in Eq.~\eqref{Fisher}, the minimization of $\nor{F^{-1}}^2$ is achieved when the set $\{\mathbf{g}_i\}$ forms a tight frame. Tight frames are minima of the so called \textit{frame potential} (see \cite{casazza2006}), calculated as:
\begin{equation}\label{FP}
FP(\{\mathbf{g}_i\})= \sum_{ij=1}^{N}|\mathbf{g}^T_i\mathbf{g}_j|^2.
\end{equation}
Let us calculate the frame potential in our case.
Here we denote with $W$ the matrix whose columns are the vectors $\mathbf{w}_i$, simple cells receptive fields. Hence
the response matrix (i.e. the simple cells output) will be $A=X^TW$ where $X$ is the dataset corresponding to the initial stimuli. The complex grid cells will then aggregate some of the responses, i.e. they will calculate $AJ$ where $J$ is a vector of zeros and ones selecting which simple cells are meant to aggregate (we will have a zero whether the simple cell is not selected in the aggregation process and one elsewhere).
We can now use this notation to write the Fisher information as follows:
\begin{equation}\label{final}
F=-\Big(\frac{\partial \mathbf{r}(\mathbf{y})}{\partial \mathbf{y}}\Big)^{\dagger} \frac{\partial \mathbf{r}(\mathbf{y})}{\partial \mathbf{y}}=J^T\dot{A}^{T}\dot{A}J=J^TRJ.
\end{equation}
with the dot indicating the derivative and $R=\dot{A}^{T}\dot{A}$. In order to minimize the number of simple cells pooled by the complex cell, we add a sparsifying term in $\nor{J}_{0}$ or its relaxation $\nor{J}_{1}$.
Given the above reasoning our minimization problem boils down to:
\begin{equation}\label{algo_GD}
\argmin_{J}\;\; \nor{J^TRJ}^2 + \lambda \nor{J}_1
\end{equation}
To find the solution we adopt a gradient descent strategy with shrinkage: a calculation shows that the update rule for $J$ is:
\begin{equation}\label{gradJ}
J\rightarrow thr(J-\lambda\;RJJ^TRJ)
\end{equation}
where the $thr$ threshold is enforcing the sparsity constraint.

\pagebreak
\section*{Authors' contributions}
F.A conceptualized the problem. F.A. and B.F. developed, implemented, and tested the model. F.A , B.F and M.M.M. wrote the manuscript.

\section*{Acknowledgements}
M.M.M. is supported by grants from the Swiss National Science Foundation (grant 320030-169206), the Fondation Asile des aveugles, and a grantor advised by Carigest.
B.F. is supported by the Fondation Asile des aveugles. 
F. A. acknowledges the Center for Brains, Minds and
Machines (CBMM), funded by NSF STC award CCF-1231216 and the financial support of
the AFOSR projects FA9550-17-1-0390 and BAA-AFRL-AFOSR-2016-0007 (European Office of Aerospace
Research and Development), and the EU H2020-MSCA-RISE project NoMADS - DLV-777826.
FA is also supported by the Italian Institute of Technology.

\bibliography{biblio}

\bibliographystyle{ieeetr}

\setcounter{section}{0}
\setcounter{equation}{0}
\setcounter{theorem}{0}
\setcounter{proof}{0}

\end{document}